\documentclass{article}

\usepackage{caption}
\usepackage{color}
\usepackage[T1]{fontenc}
\usepackage{amsmath,amsthm,amsfonts,amssymb}
\usepackage[mathscr]{eucal}
\usepackage{indentfirst}
\usepackage{graphicx}
\usepackage{exscale,relsize}
\usepackage{multicol}
\usepackage[dvips,final]{epsfig}
\usepackage{bbm}
\usepackage{siunitx}
\usepackage{comment}

\theoremstyle{plain}
\newtheorem{theorem}{Theorem}[section]

\newtheorem{remark}{Remark}[section]
\newtheorem{lemma}{Lemma}[section]

\definecolor{orange}{rgb}{1,0.3,0}
\definecolor{purple}{rgb}{1,0,1}

\newcommand{\approptoinn}[2]{\mathrel{\vcenter{
  \offinterlineskip\halign{\hfil$##$\cr
    #1\propto\cr\noalign{\kern2pt}#1\sim\cr\noalign{\kern-2pt}}}}}

\theoremstyle{definition}

\def\P{{\mathbb P}}
\def\N{{\mathbb N}}
\def\E{{\mathbb E}}

\def\beqn{\begin{equation}}
\def\eeqn{\end{equation}}
\def\ol{\overline}

\begin{document}

\title{Metastability and Multiscale Extinction Time on a Finite System of Interacting Stochastic Chains}

\author{{L. Brochini}  
\thanks{Instituto de Matem\'{a}tica e Estat\'{i}stica, Universidade de S\~{a}o Paulo.}
\and {M. Abadi}
\thanks{Instituto de Matem\'{a}tica e Estat\'{i}stica, Universidade de S\~{a}o Paulo. email: leugim@ime.usp.br} 
\date{}
}

%

\maketitle

\begin{abstract}

We studied metastability and extinction time of a finite system with a large number of interacting components in discrete time by means of analytical and numerical investigation. The system is markovian with respect to the potential profile of the components, which are subject to leakage and gain effects simultaneously. We show that the only invariant measure is the null configuration, that the system ceases activity almost surely in a finite time and that extinction time presents a cutoff behavior.  Moreover, there is a critical parameter determined by leakage and gain below which the extinction time does not depend on the system size. Above such critical ratio, the extinction time depends on the number of components and the system tends to stabilize around a unique metastable state. Furthermore, the extinction time presents infinitely many scales with respect to the system size.

\end{abstract}

\section{Introduction}

The extinction time of a stochastic process in finite systems has been an issue of interest in particle systems, especially in models for epidemics, such as the contact process \cite{liggett1999,mountford1999,valesin2010,mountford2016,schapira2017}. The route to extinction is known to be affected by the emergence of metastability \cite{schonmann1998,cassandro1991,mountford2013,mourrat2016}, a scenario set by phase transitions exhibited by the corresponding infinite system. For the contact process in finite graphs, for instance, the system is subject to an abrupt change in behavior according to a certain parameter, the infection rate, being below or above a critical value, causing the extinction time to depend logarithmically on the number of particles in the subcritical case \cite{chen1994} and exponentially in the supercritical case \cite{cassandro1984,mountford1999}. 

Here we investigate the phenomenon of metastability and examine the extinction time with respect to the network size in a different system of interacting stochastic chains. The study is based on analytic calculations explored in by numerical experiments. Our model is inspired by leaky integrate and fire neuronal networks. A stochastic version in discrete time of such models can be found in \cite{Brochini2016,Costa2017}  which is a special case of the model introduced in \cite{Galves2013}. 
Here we deal with a modified version where the components are under some kind of environmental competition to fire. 

The model is as follows. The system is composed of $N$ components, naturally identified as $1,2,\dots,N$. Each component $i$ has a potential that evolves in discrete time, given by $ U_n(i)$ at time $n$. We denote as ${\bf U}_n$ the potential array ${\bf U}_n=( U_n(1),\dots, U_n(N))$. The system evolution is as follows. It is an order one Markov chain where the state of the system ${\bf U}_{n+1}$ at time $n+1$,  depends probabilistically only on the state of the system in the previous instant ${\bf U}_n$.

The system, at every time step, is subject to a leakage effect that leads to a reduction in the potential of each element by a factor of $0<\mu<1$. Moreover, it may either exhibit spontaneous discharge of one -- and only one -- of its components or no discharge at all. If a component discharges at a given time, then its potential is reset to zero in the next instant while all other components suffer a uniform and positive increase of potential $w$.  The probability that a certain component will discharge is a saturated linear function of its potential at that time, normalized over the system size.

Since after firing the potential of a component resets to zero, the probability of the system not to fire is always positive and bounded from below by $1/N$. 
It will cause the system to eventually cease to fire. In order to understand how the system behaves up to this time, it is useful to examine an associated stochastic process $\{Y_n\}_{n\in \mathcal{N}}$ that represents the system firing history, where $Y_n$ is an indicative function that any component fired at time $n$. The process extinction time is then equivalent to the last time a 1 is observed in the process $Y$. Note that the system potential does not achieve null configuration upon extinction. Instead, it approaches asymptotically the absorbing state. Therefore, we have to deal with the issue that the extinction time is not a stopping time, which is addressed in sections \ref{nofire} and \ref{StopCriterion}.

The main result of this paper is that a parameter dependent on leakage and gain  determines a transition on the extinction time: below a critical value of this parameter it does not depend to the system size and above the critical value it presents infinitely many scales. We found that the extinction time law with respect to the number of components undergoes multiple transitions in this system. An important factor to this matter is the ratio $\gamma=w/(1-\mu)$. We show that if $\gamma<1$ the extinction time does not depend on the number of components. Conversely, if $\gamma>1$, the potential array tends to uniformize. Under the firing regime it approaches an invariant potential close to $\gamma$, which configures a metastable state.  In this case, the firing blocks follow a geometric distribution of parameter $1/N$. 

Surprisingly, still when $\gamma>1$, the extinction time law not only depends on $N$, but suffers multiple changes according to the specific values of $\mu$ and $w$.  Considering the situation that the system is close to the metastable state, the route to extinction is ruled by how the system loses potential until it reaches a value below one. We found a crucial role for the number $m$ of times the system must fail to fire to reach a potential below 1: the extinction time is proportional in expectation to $N^m$.  Note that $m$ is independent of $N$ and also determined by the specific values of $\mu$ and $w$, not only the ratio $\gamma$, meaning there are systems that reach the same metastable average potential close to the same value of $\gamma$, having very different extinction time scales.  Moreover, for any given value of $w$ there can be infinitely many changes in the extinction time law for $\mu$ in the interval $]0,1[$. These results are presented in sections \ref{regimes} and \ref{TauSec} and illustrated numerically in section \ref{SecSimul}.

\section{The Model}

Let $N\in \N$ be the number of components of the system. We define a Markov chain
chain ${\bf U}_n= (U_n(1),\dots,U_n(N))_{n\in\N}$, where $U_n(i) \ge 0, \  i=1,\dots,N$.  The evolution of the system will be determined by the following two parameters. Let $0<\mu<1$ be the leakage factor and  $w>0$ be the potential gain of a non-firing component after the system fires. At each time step, a component $i$ may fire. We also introduce the resulting potential vector $\Delta^i(u)$ after the component $i$ of the given vector $u$ fires. Then denote
\[
\Delta^i (u)(j)=   \left\{
\begin{array}{ll}
 \mu \ u(j) +w  &\  \  j\not= i ,\\
0 &\ \  j=i .
\end{array}
\right.
\]

This means that each component that  did not fire suffers a leakage  effect  having now a proportion $\mu$ of its  potential on the previous instant while also receiving an increase of $w$ due to the stimulus of the component that fired.
The component that fired resets its potential to zero.
If the system did not fire, then there is only the leakage effect over all the components.
The Markov chain is defined by the transition probabilities given as follows:

 \begin{eqnarray*}
\P( {\bf U}_{n+1}= \Delta^i(u)\ | \ {\bf U}_n=u) &=& \frac{\phi(u(i))}{N} , \quad i=1,\dots,N, \\
\P( {\bf U}_{n+1}= \mu u \ | \ {\bf U}_n=u)   &=& 1-  \frac{\sum_{j=1}^{N}\phi(u(j))}{N} .
\end{eqnarray*}

Where the firing probability function $\phi:{\mathbb R}_{\ge 0}\to [0,1]$ is the truncated identity
\[
\phi(u)= \left\{
\begin{array}{ll}
u &   0 \le u<1 ,\\
1 & u \ge 1 .
\end{array}
\right.
\]


Since $\phi \le 1$, the above transition probabilities are well defined.\\

\section{Mean potential evolution}

Our results will show that the process has a unique invariant state that is the zero configuration ${\bf U}=(0,\dots,0).$ 
This state is reached only when taking limit to infinity in the time scale.
However, we will show that the system ceases activities almost surely at  a finite time $\tau$.

We can describe the evolution of the process until $\tau$ in part with the evolution of the total potential of the system at time $n$.
When the system does not fire at time $n+1$, all elements lose potential resulting in the equation ${\bf U}_{n+1}=\mu {\bf U}_n$ 
smaller than the previous total potential ${\bf U}_n$. 
On the other hand, when the system fires, one can ask what is the invariant total potential, that is the solution to the equation

\[
\sum_{i\in I}U_{n+1}(i) = \mu  \sum_{i\in I \setminus \{i_0\}}U_{n}(i) + (N-1) w  \ ,
\]

where $i_0$ is the spiking site and $I$ is the set of all components. Even when the actual potential depends on the firing site $i$, we show later on that
the system has a tendency to uniformize  their  potentials 
over the sites  (except for the last spiking one) 
and thus the no-null potential ${\bf U}_n(i)$ may be well approximated by
 $U_n/(N-1)$. 
Under this condition, the \emph{invariant} total potential under the firing regime is
\[
\mathcal{U}=\mathcal{U}(\mu, w,N) = (N-1)  \frac{w}{1-\mu - \frac{\mu}{N-1}} \ .
\]
As a consequence,  the typical non-zero potential (during this sustained firing regime) in a given site and for large $N$ is approximately 
${w}/{(1-\mu)}$. 
Direct computations show that for a total potential above  $\mathcal{U}$,  the system loses total potential whether it discharges or not. 
On the contrary, for a potential   below $\mathcal{U}$, the system typically loses total potential whenever it fails to discharges and gains total potential whenever it discharges. Moreover, even when the system discharges, one typically gets  a potential below $ \mathcal{U}$.
This makes the region of potential levels above $\mathcal{U}$ a transient one.
Once the system is close to $\mathcal{U}$, it either loses or gains potential, moving along the interval $(0, \mathcal{U})$.

The next lemma quantifies   the system tendency to uniformize their potentials mentioned above.
To that,  we define an auxiliary sequence of random vectors $V_n=(V_n(1),\dots,V_n(N))$ with entries $V_n(i)$ given by the \emph{order statistics} of $U_n(i)$. Namely, define
\[
V_n= ({U}_n^{(1)},\dots,{ U}_n^{(N)}) \ ,
\]

where ${ U}_n^{(1)} \le { U}_n^{(2)} \le \dots \le { U}_n^{(N)}$ is a re-ordering of the components of ${\bf U}_n$.

The next lemma says firstly that the non-firing effect keeps  the ordering of the potential  while 
 the firing effect keeps the ordering for the potentials, except for the firing component that resets to zero and its potential becomes the 
 smaller one.

 Let $Y_n$ 
be the indicator  function that  the system fires at time $n$. \\

\begin{lemma} 
\begin{itemize}
\item[(a)]  Suppose $Y_{n+1}=0$, then $V_{n+1}(i)=\mu V_{n}(i)$ for all $i=1,\dots,N$.
\item[(b)]  Suppose $Y_{n+1}=1$ and that the index of  $V_n$  corresponding to the  firing component is $i_0$.
Then $V_{n+1}(1)=0 $, 
$ V_{n+1}(i)=\mu  V_{n}(i) + w$, for all $i=i_0+1,\dots,N$ and 
$ V_{n+1}(i)=\mu V_{n}(i-1) + w$, for all $i=2,\dots,i_0$.
\end{itemize}
\end{lemma}

\begin{proof}
Item $(a)$ follows since the application $u\to \mu u$ is monotonic. 
Item $(b)$ follows since the application $u\to \mu u + w $ is also monotonic.
\end{proof}

\begin{remark} Item $(b)$ says in particular that, given $U_n$ (or equivalently $V_n$), the way the system receives more potential
is when $i_0$ equals 2, namely the firing component is the one with the minimum potential and the way the system receives
less potential is when the firing component is the one with the largest potential.  
\end{remark}

There are two effects that work concomitantly to make the system evolve towards uniformity. First, the influence of the firing site over all the others is always $w$, regardless of the potential it had at firing time. Additionally, all elements lose a portion of their potential at rate $\mu$. Both effects are responsible for simultaneously attenuating very large potentials and increasing very low potentials. 
The lemma below   
shows the stable potential

\begin{lemma} \label{basin}
Suppose  $ (1-\mu^{i-1}) \ \gamma    \le V_n(i) \le \gamma$ for all $i$ and that $Y_{n+1}=1$. Then
$  (1-\mu^{i-1}) \ \gamma    \le V_{n+1}(i) \le \gamma . $          
\end{lemma}

\begin{proof}
$Y_{n+1}=1$ means that the system fired. Thus, $\phi(V_{n+1}(1))=0$.
Further, suppose that the component that fired  was $i_0$.
Now, observe that $\gamma$ is a fixed point for the transformation $u\to \mu u + w$.
Thus, by the lemma above, for $i\not=i_0$, one has 
\[
V_{n+1}(i)= \mu V_{n}(i)+w \le \mu \gamma +w= \gamma.
\]
This shows the second inequality.
Now consider $i\not=i_0$. 

Thus, for $i>i_0$ we have

$V_{n+1}(i) =  \mu V_{n}(i)+ w \ge \mu \gamma (1-\mu^i) +w  =   \gamma (1-\mu^{i+1}) >  \gamma (1-\mu^i) $.
For $ i<i_0$ 
a similar computation holds: 
$V_{n+1}(i) =  \mu V_{n}(i-1)+ w \ge \mu \gamma (1-\mu^{i-1}) +w   =   \gamma (1-\mu^{i}) $.

\end{proof}

 The lemma motivates to introduce the following set:
$$
{\bf B}=\{x_1^ N \in \mathbb{R}^ N: (1-\mu^{i-1}) \ \gamma \le x_i \le \gamma,  \forall 1\le i\le N \}
$$ 
The result of the lemma states  that ${\bf B}$ acts as an invariant state for the firing regime. 
Briefly, if $Y_{n}=1$,  then ${\bf U}_{n+1}\in {\bf B} $, which defines the  meta-stable state of the system.

\section{The three regimes} \label{regimes}

To describe the evolution of the process ${\bf U}_n$ it is useful to describe the evolution of the fire/non-fire process $Y_n$.
A realization of the process $Y_n$ can be described as a composition of sequences of three regimes: the firing regime, the non-firing regime
and the mixed regime.
The first one corresponds to a continuum of discharges (fires)  which happens when the system is, typically,  in the meta-stable state ${\bf B}$.
The second regime begins with the first fail to fire until the next fire or to infinity if the  system ceases to fire.
In the last case, the third regime begins. It lasts until the moment that the system recovers a certain minimum potential 
${\bf U}_n$ such that $V_n(i)\ge \frac{1-\mu^{i-1}}{1-\mu}L$, with
$1\le L\le \gamma$, or to infinity otherwise. The level $L$ is reached when the system behaves, in some sense, similarly as it does when in ${\bf B}$. It will be defined precisely later on.
This, together with other analytical characteristics are described below.

\subsection{ Firing  blocks    }

To describe the law of the  firing  regime we denote with $\theta$ an upper and a  lower bound for the probability 
that the system fires, given that  $\bf{U}_n \in {\bf B}$. That is
\[
\theta= \phi(\gamma)   \left( 1-\frac{1}{N} \right) ,
\]
and
\[
\eta=\phi(\gamma) \left( 1-        \frac{1}{ N} \  \frac{ 1- \mu^N }{1-\mu} \right) \  .
\]
Both bounds follow directly from integrating the bounds given in Lemma \ref{basin}.
Note that 
\[
0 <  \theta-\eta  =  \frac{ \phi(\gamma) \mu^N}{(1-\mu)N} \le  \frac{  \mu^N}{(1-\mu)N}  , 
\]
which shows the closeness between both bounds for large $N$.

The next lemma gives a full picture of the statistical behavior of the firing blocks. 
By the Markovian property, it is enough to describe a firing block beginning at the origin.
The same will be done later on for the non-firing ones.

\begin{lemma} \label{fire}
\begin{itemize}
\item[(a)] Markovian type property 
\[
\eta
 \le 
\P(Y_{n+1}=1 |  Y_n=1, {\bf U}_{n}\in {\bf B}, Y_{n-k}^{n-1}=y_1^{k} )  
\le  \theta .
\]
\
\item[(b)] Geometric fire regime. Let ${\bf U}_{0}\in {\bf B}$. $T=\max\{n:  Y_n=1\}$ verifies
\[
 \eta^t
\le
\P(T > t) \le  \theta^t \quad \forall t \ge 1 .
\]
\end{itemize}
\end{lemma}

\begin{proof}
To prove $(a)$
we   have
$\phi(V_{n}(1))=0$
and  the function $\phi$ also determines that  $ \phi(V_{n}(i)) \le  \phi(\gamma) $. 
Thus, the upper bound in $(b)$ follows summing up along $i=2,...,N$ and dividing by $N$.
Similarly  $ \phi(V_{n}(i))\ge \phi( \gamma(1-\mu^i))  =   \min \{ \gamma(1-\mu^i), 1 \}  \ge \phi(\gamma) (1-\mu^i) $.
We also get the lower bound summing up  along $i$. 

To prove $(b)$ consider
\[
\P(T> t)  =  \prod_{i=1}^{t}  \P( T >i |  T> i-1 )   \P(T >0) .
\]
The last factor is equal to one by definition of $T$.
Each factor in the product verifies
\[
 \P( T >i |  T> i-1 ) =
\P( Y_{n+1}=1 |    Y_{ n-i }^{n-1}={\bf 1} ) .
\]
Using  $(b)$ we finish the proof.
\end{proof}

The last statement of the previous lemma establishes bounds for the expected length of the firing blocks,
$(1-\theta)^{-1} \le \E(T)\le (1-\eta)^{-1}$. Now, when $\gamma\ge 1$, one has $(1-\theta)^{-1}= N$ while for the $\gamma < 1$ case (considering large $N$) 
$(1-\theta)^{-1}\approx (1-\gamma)^{-1}$.
The corresponding values of $(1-\eta)^{-1}$ are close to the previous ones.

This fact motivates us to name $\gamma_c=1$ a critical value of the parameter $\gamma$ for large $N$. In fact, we  refer to the case $\gamma<\gamma_c$ as subcritical in which case the size of a firing block is independent on  $N$. Conversely, we refer to  $\gamma>\gamma_c$ as the supercritical case where the firing activity is sustained for a geometric time of parameter $1/N$, during which  the system remains in a metastable state where  the average potential is very close to $\gamma$.

\subsection{ Non-firing blocks} \label{nofire}

Different from the firing regime who has a close to Markov behavior, the non-firing regime is ruled  by a property  close to the
renewal one. That means that one has to look back until the last fire of the system, and the distribution of the
non-firing blocks depends on how long this last discharge happened.

Yet, in the super-critical case,  
we distinguish two sub-regimes. One occurs since the system stopped firing and up to having a potential lower than one.
This depends, therefore, in $\gamma$ and is due  to $m=\inf\{k \in  \mathbb{ N}\ | \  \mu^k\gamma<1 \}$.
In this regime, by the shape of the  function $\phi$, the probability of non-firing keeps being the same and equals to $1/N$.  
The second sub-regime begins when the potential gets below one and there is no uniformity of the 
non-firing probabilities, they depend on the potential level itself.

The following lemma is the counterpart of Lemma \ref{fire} 
in the non-firing  regime.
It provides bounds for the potential ${\bf U}_n$ itself and for the probability that the system fires
given that there are exactly $k$ consecutive
times the system did not fire in the immediate past.

\begin{lemma} \label{nonfire}

Suppose ${\bf U}_{n-k} \in {\bf B}$ 
\begin{itemize} 
\item[(a)] Suppose  $Y_{n-k+1}^{n}={\bf 0}$ (empty set in case $k=0$), $Y_{n+k}=1$. Then  ${\bf U}_{n+1} \in \mu^ k {\bf B}$

\item[(b)] Let $m=\lceil \frac{\log(1/\gamma)}{\log \mu} \rceil $.  Then
\begin{itemize}
\item for $k < m$ 
\beqn \label{pqx1}
   \eta \le \P(Y_{n+1}=1 |  Y_{n-k+1}^{n}={\bf 0}, Y_{n+k}=1, {\bf U}_{n-k}\in {\bf B}) \le  \theta .
\eeqn
\item for $k \ge m$ 
\beqn \label{pqx2}
  \mu^{k} \eta 
 \le 
\P(Y_{n+1}=1 |  Y_{n-k+1}^{n}={\bf 0}, Y_{n+k}=1, {\bf U}_{n-k}\in {\bf B}) \le \mu^k \theta .
\eeqn
\end{itemize}
\end{itemize}

\end{lemma}

\begin{proof}
This proof  follows mutatis mutandis the proof of  Lemma \ref{fire}.
\end{proof}

The following lemma describes the behavior of the length of each non-firing block  of the system 
\emph{once the potential gets below one}. As before for the firing one, we set the origin of the non-firing block at the origin.

\begin{lemma} \label{bonfire}
Let ${\bf U}_{0}\in {\bf B}$ and define $S=\max\{ \max\{n:  Y_n=0\}-m,0\}$.
Let 
$$
K=(1-   \mu^{m}\theta)^{\frac{-1 }{1-\mu}}, \qquad \text{and}   \qquad
J= (1-\mu^{m}\eta)^{   (1-\eta)  \frac{-1}{1-\mu}  } .
$$
Then
\begin{itemize}
\item[(a)] The probability  that the system  completely ceases to  fire  at any non-firing block verifies
\[ 
   \frac{1}{K}
 \le  \P( S=\infty   )  \le    \frac{1}{J} .
 \]
\item[(b)] 
The conditional distribution of $S$ given that it is finite verifies 
\[  
 \frac{  K^{-\mu^t }  -   1 }{K - 1 }
\le  \P( S \ge t |   S<\infty)  
\le  \frac{  J^{-\mu^t }  -   1 }{J - 1 } .
\] 
\end{itemize}
\end{lemma}

\begin{proof} 
By definition of $S$ 
\[
 \P(S \ge t) =  
  \prod_{k=m}^{t-1} \P(S \ge k+1 \  | \ S \ge k)  , 
\]

and
\begin{eqnarray*}
 \P(S \ge k+1 \  | \  S \ge k) 
& =&1-  \P(Y_{k+1}= 1 |  Y_{1}^{k+m} = {\bf 0}, U_{k+m-1}\in {\bf B} ) .\\ 
 \end{eqnarray*}
By 
(\ref{pqx2}), the above display is lower and upper bounded  respectively by
\[
 1-\mu^{m+k}\theta , \qquad  \text{and} \qquad       1-\mu^{m+k}\eta .
\]
The first expression  is bounded from above by
$(1-\mu^{m}\theta)^{\mu^{k}} $.
We used the inequality $1-xy \le (1-y)^{(1-y)x}$ that holds for all $0<x<1, 0<y<1 $.
Summing over $k$ we get
\[ 
 \P(S \ge t)  \ge (1- \mu^m\theta)^{  \frac{1 -\mu^{t}}{1-\mu}  }\ .
\]
which converges to the claimed lower bound and proves the first inequality in (a).
Similarly, to get the upper bound we also use  
(\ref{pqx2})
to get that   $\P(S \ge t) $    is bounded from above by
\[
 (1-  \mu^m\eta)^{(1-\eta) \frac{1 -\mu^{t}}{1-\mu}} \ .
\]
Here we used the inequality $1-xy \le (1-y)^{(1-y)x}$ for all $0<x<1, 0<y<1 $.
Taking limit on $t$, we get the upper bound. \\

The proof of (b) follows by normalization  
to get a tail distribution
\[
 \P(S \ge k \ | \ S < \infty)  =
 \frac{ \P(S \ge k ) -\P (S = \infty )    }{  1- \P(S = \infty)} .
 \]
\end{proof}

\begin{remark}
The last item of the previous lemma says that the time elapsed during a finite, non-firing regime, is independent of $N$
and so it contributes only with a constant (depending on $\mu$ and $w$ to the extinction time).
Moreover, this constant may be  large only when  $\mu$ is close to one.
Moreover, the conditional expectation verifies
$
\mu^{-1} \le  \E(S \ | \ S < \infty)  \le  \mu^{-2}   . 
$
\end{remark}

\begin{remark} \label{level}
To enter the regime described in the second part of Lemma \ref{nonfire} $(b)$, the system needs to cross the strip $[\gamma,1]$,
which takes $m$ consecutive steps.
This fact is determinant for the length of the extinction time of the system.
During a period of recovery, actually take a very long time 
for the system potential  to get again  ${\bf U}_n \in {\bf B}$. However, there is the minimum level $L= \min\{u>0 \ | \  \mu^m u= 1 \}$ which also needs $m$ steps to cross this strip. This level can be reached after a relatively short sequence of fires.
\end{remark}

\subsection{The mixed regime: almost martingale properties} \label{mixed}

In this section we consider the third regime, in which are observed some firing times but  they are not enough to recover,
and also some non-spiking periods but they are also not enough for the system  to stop activities.
That is, the potential of the system fluctuates between the minimum recovery potential level $L$  defined in Remark \ref{level} and a  low  level until it 
fully recovers or it ceases activities.

Here we show that for the critical regime, either the mean potential level or the total potential have a behavior close to a martingale.
For the supercritical case, the behavior is close to a sub-martingale and for the sub-critical, close to a supermartingale.

Consider first a total potential ${\bf U}_n>0$ such that   $U_n(i) \le 1$ for all $i$ and let us  compute the expectation for the next mean potential level,
\begin{eqnarray*}
\E(\ol U_{n+1}| \ol U_n) 
&=& (\mu \ol U_n + w -  \frac{\mu U_n(i_n)+w}{N} ) \ol U_n 
+ \mu \ol U_n (1-  \ol U_n)  \\
&= &   \left( \mu   +w  -  \frac{\mu U_n(i_n)+w}{N} \right) \ol U_n  .
\end{eqnarray*}
where $\ol U_n= \sum_{i=1}^{N}/N$ and $i_n$ denotes the spiking site at time $n$.
Take $\Delta =\frac{\mu U_n(i_n)+w}{N}$.
Since we are considering the regime where $U_n(i)\le 1$ for all $i$, 
$\Delta$ becomes bounded by $\gamma/N$. 
 We conclude that the mean potential verifies
 \[
| \E( \ol U_{n+1}|\ol U_n)  -  (\mu   +w) \ol U_n  | \le     \frac{\gamma}{N}   .
  \]
This can be interpreted as (except for a small random fluctuation bounded by $\gamma/N$) $\ol U_n$ being a martingale, super-martingale and sub-martingale according to the critical, sub-critical and super-critical regime respectively.

We now compute the variance. To that  
\begin{eqnarray*}
 \E( \ol  U^2_{n+1}   | \ol U_n) -  \E^2( \ol  U_{n+1}   | \ol U_n) 
=  (w-\Delta)^2 \ol U_n (1-\ol U_n) .
\end{eqnarray*}

This means that, as the potential approaches zero or one, the system becomes increasingly more biased towards firing or not firing respectively. This means that the system becomes more deterministic and has a major tendency to either recover or die as each approaches one or zero respectively. We can heuristically interpret that the system does not spend a significant amount of time in this regime.

\section{Extinction Time}\label{TauSec}

Let $\tau$ be the time to the last fire of the system
\[
\tau= \max\{n \ge 1 \ | \ Y_n=1 \} .
\]
The next  result states the lifetime of the process is almost surely finite.

\begin{theorem}
The process is almost surely finite. Namely
\[
\P( \tau < \infty) =1 .
\]
\end{theorem}
\begin{proof}
The above probability is lower-bounded by the probability of having a non-firing block of infinite size.
The non-firing blocks are measured after the system profile ${\bf U}_n$ gets to an order profile such that $V_n(i)\ge \frac{1-\mu^i}{1-\mu}L,$
(the recovery level in which the system behaves as in {\bf B})
and then starts to fail. 
The Markovian property of the firing blocks guarantees that the size of the non-firing blocks are independent. Since the probability of a non-firing block to be infinite is positive, Borel-Cantelli lemma says that the claimed lower bound has probability one.
\end{proof}

Now we address the statistical properties of  the extinction time $\tau$ with respect to the size of the system.
The properties depicted in section \ref{regimes} indicate that a  contribution to $\tau$ is given by the firing blocks,
which are of mean size  $N$ for the supercritical case and independent of $N$ for the subcritical case. Furthermore, the other two
regimes are independent of $N$.  So the question is how many firing blocks one typically observes. 

Recall that the firing block is interrupted by a non-firing one, and it has a statistically distinct regime at its very beginning: for a finite time not larger than $m$ it preserves the evolution properties of the firing regime, in the sense that it preserves the probabilities
of firing/non-firing until it crosses the mean potential threshold equal to one.
The first one tends to make very short excursions below the meta-stable state ${\bf B}$ due to the probability $1/N$ of not firing
 (and $N-1/N$ to fire). Thus, the system attempts to cross the critical potential level equal to one, \emph{once it ceases to fire}.
The probability to cross this level is, at least $N^{-(m-1)}$.
These two contributions make the time $\tau$ bounded from below by a mean of $N^{m}$.

The values of $m$ determine all possible different scales of the extinction time. Intuitively, one may think of a simplified version of the system. Beginning at the meta-stable state, the system starts a firing block  and shall attempt to cross the critical potential level with probability bounded by $1/N^m$.  The number of attempts follow a geometric law given by the probability to recover the potential up to $\frac{1-\mu^i}{1-\mu}L$ (during the non-firing or mixed regimes), which does not depend on the system size, as shown in sections
 \ref{nofire} and \ref{mixed}. Therefore, the  observed extinction time  is  ruled  by $N^m$.

Interestingly the exponent $m$ has  non-trivial relation with parameters $\mu$ and $w$. Solid lines in Fig.~\ref{ExpIdeal} show  there are infinitely many transition lines for crescent integer $m$ values in the plane $\mu\times w$. Straight dashed lines in Fig.~\ref{ExpIdeal} indicate the $(\mu,w)$ relation for fixed $\gamma$ values. There is only one possible $m$ value when $\gamma<1$ (green line) that is $m=0$. However, when $\gamma>1$, the lines for constant $\gamma$ cross the $m-$curves as $\mu$ increases, meaning that two systems that reach the same metastable state can present very different extinction time scales. Note that the line for $\gamma=1$ coincides with the magenta curve that defines $m=1$.

The next sections present empirical results of the behavior of the extinction time with $m$ and $\mu$.

\noindent%
\begin{minipage}{\linewidth}
\makebox[\linewidth]{
\includegraphics[width=1\columnwidth]{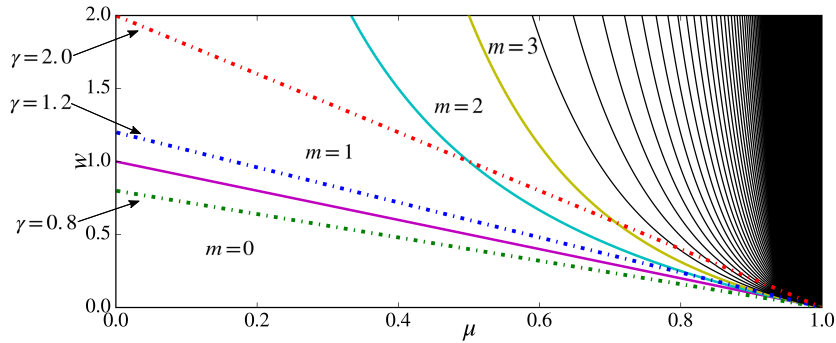}}
\captionof{figure}{ \textit{ Relation of the integer exponent $m$  with parameters $w$  and $\mu$.}  Dashed lines correspond to fixed values of $\gamma=0.8,1.2$ and $2.0$ with colors green, blue and red respectively. Solid lines represent the limit when $m$ values change. The area below the magenta line corresponds to $m=0$; the area between the magenta and cyan lines corresponds to $m=1$; the area between the cyan and yellow line corresponds to $m=2$, and so forth. }

\label{ExpIdeal}
\end{minipage}

\section{Simulations}\label{SecSimul}

\subsection{Process Evolution}

Here we illustrate the model behavior by depicting the evolution of the mean potential in time for different simulated cases. In the subcritical case (red line on Fig. \ref{SampleEvol}), the mean potential tends rapidly towards zero and the system does not present metastability. On the other hand, for the supercritical case (blue line on Fig. \ref{SampleEvol}A ), the system tends to keep firing and stabilizes its potential towards $\gamma$. The system fails to fire for the first time at $n=69$, but because there is still a high probability of firing again, the system resumes firing and rapidly recovers towards $\gamma$.  Fig. \ref{SampleEvol} B depicts the complete time series for this simulation of the supercritical case, where a behavior of  long firing sequences with some ocasional failures can be observed. At some point, a very long sequence of approximately 150 failures is observed, which causes the mean potential (and consequently the probability to fire) to approach zero irreversibly for all numerical purposes (see Sec.~\ref{StopCriterion}). When this event is observed, we take the last time a fire is observed and call it the observed extinction time, which is  $\tau_{obs}=3139$ in this case.

A remarkable fact is that systems with different gain and leakage terms are able to reach the same quasi-stationary state $\gamma$, differing only with respect to \textit{how} the system detours from $\gamma$ (depending on $\mu$) or approaches $\gamma$ (depending on $w$) Fig. ~\ref{SampleEvol} C. Note that for the green line ($w=0.3, \mu=0.7533$) at every failure to fire (after a large enough firing block), the average potential drops to a value below one, whereas for the blue line ($w=0.15,\mu=0.878$), when the system is close to $\gamma$ it must fail to fire more than once in order to drop below the average potential one. Therefore, even though the metastable average potential is the same, the \textit{extinction time} is different for each pair of parameters ($w,\mu$). 

We examined the proportion of time that the average potential is smaller than $\gamma-\epsilon$  before time of extinction $\tau_{obs}$.
Numerically
\[
	q(\gamma,\mu)=\mathsf{mean}\{\frac{ | \{0<n<\tau_{obs}:\ol U_n <\gamma-\epsilon\} | }{\tau_{obs}} \} ,
\]
where each $\mathsf{mean}$ is calculated across $k$ simulations, summing up a total time of activity equal to $\sum_{i=0}^{k} \tau_{obs,i} =\num{5e6}$. Figure ~\ref{SampleEvol} D   shows $q(\gamma,\mu)$ still for the same $\gamma=1.2$, using $\epsilon=0.12$ calculated for 20 trials for each $N=200,300,400$ and $500$. 
We observe two concurrent effects: while $q$ increases with $\mu$ for fixed $N$, it also decreases as $N$ increases when $\mu$ is fixed.

The first effect is explained by result in lemma \ref{nonfire} b. Moreover, we predicted that the typical size of non-firing blocks, conditioned to being finite, can be arbitrarily large. This can be deduced from the conditional distribution in  lemma \ref{nonfire} b. In particular, note that in the extremal case where $\mu=1$ and $w=0$ one obtains the identity transformation which produces the indistinguishability of firing and non-firing regimes. 

The second effect can be explained by the fact that the mean duration of non-firing blocks conditioned to finiteness is independent of N (lemma  \ref{nofire} B) 
 while the firing regime increases with N -- as shown in section 4 that excursions below $\gamma$ occur after a geometric time of mean $N$ for $\gamma>1$.   This explains why the proportion $q$ decreases as $N$ increases.

Overall, Fig.~\ref{SampleEvol} D shows that the time the system spends away from $\gamma$ before extinction does not contribute significantly to the extinction time of the process.

As discussed in Sec.~\ref{TauSec}, the route to extinction is governed by the system attempts to cross the critical potential level equal to 1, causing the extinction time to be dependent on the size of the system with $N^m$. This is empirically shown in the next section.

\noindent%
\begin{minipage}{\linewidth}
\makebox[\linewidth]{
\includegraphics[width=0.5\columnwidth]{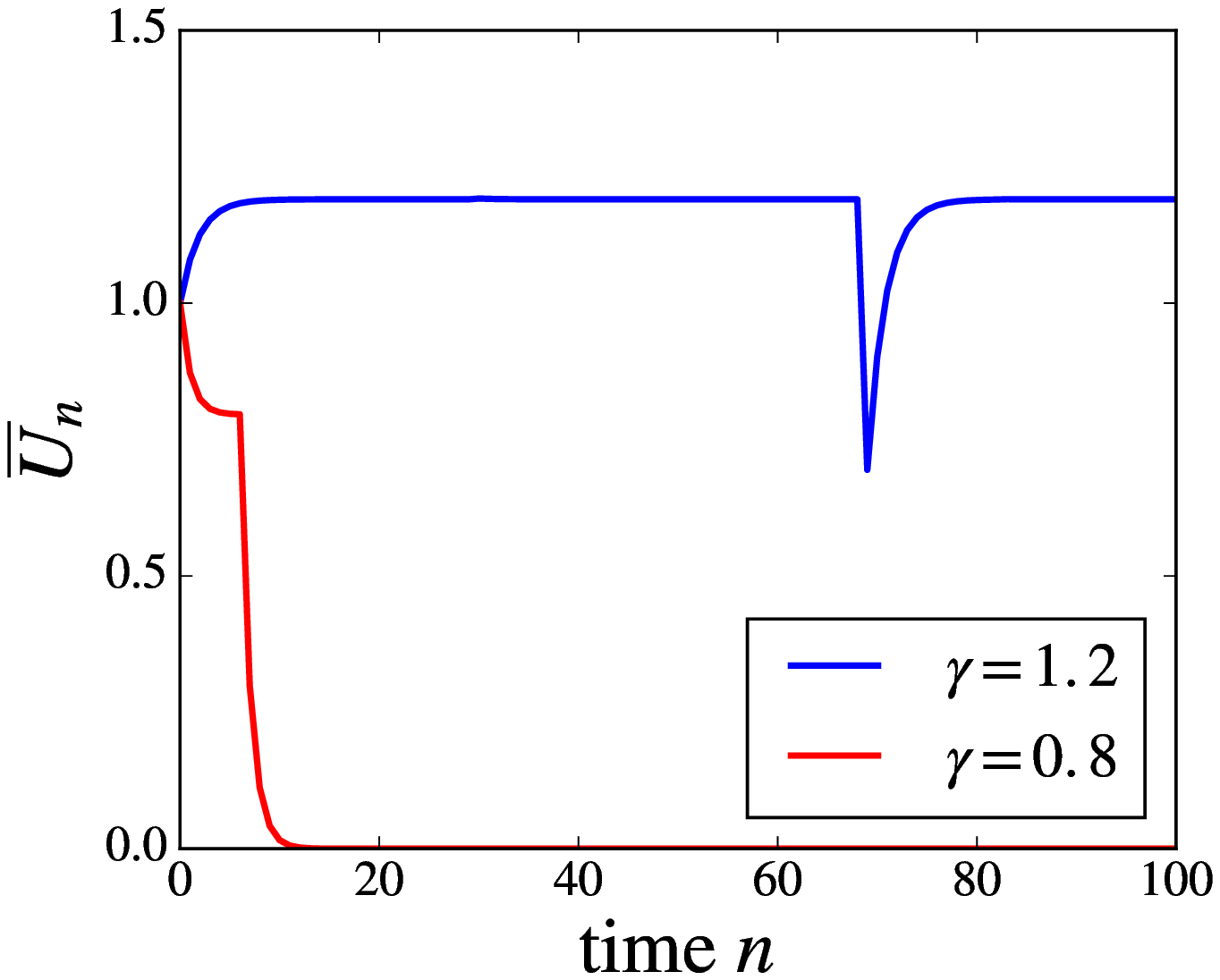}
\includegraphics[width=0.5\columnwidth]{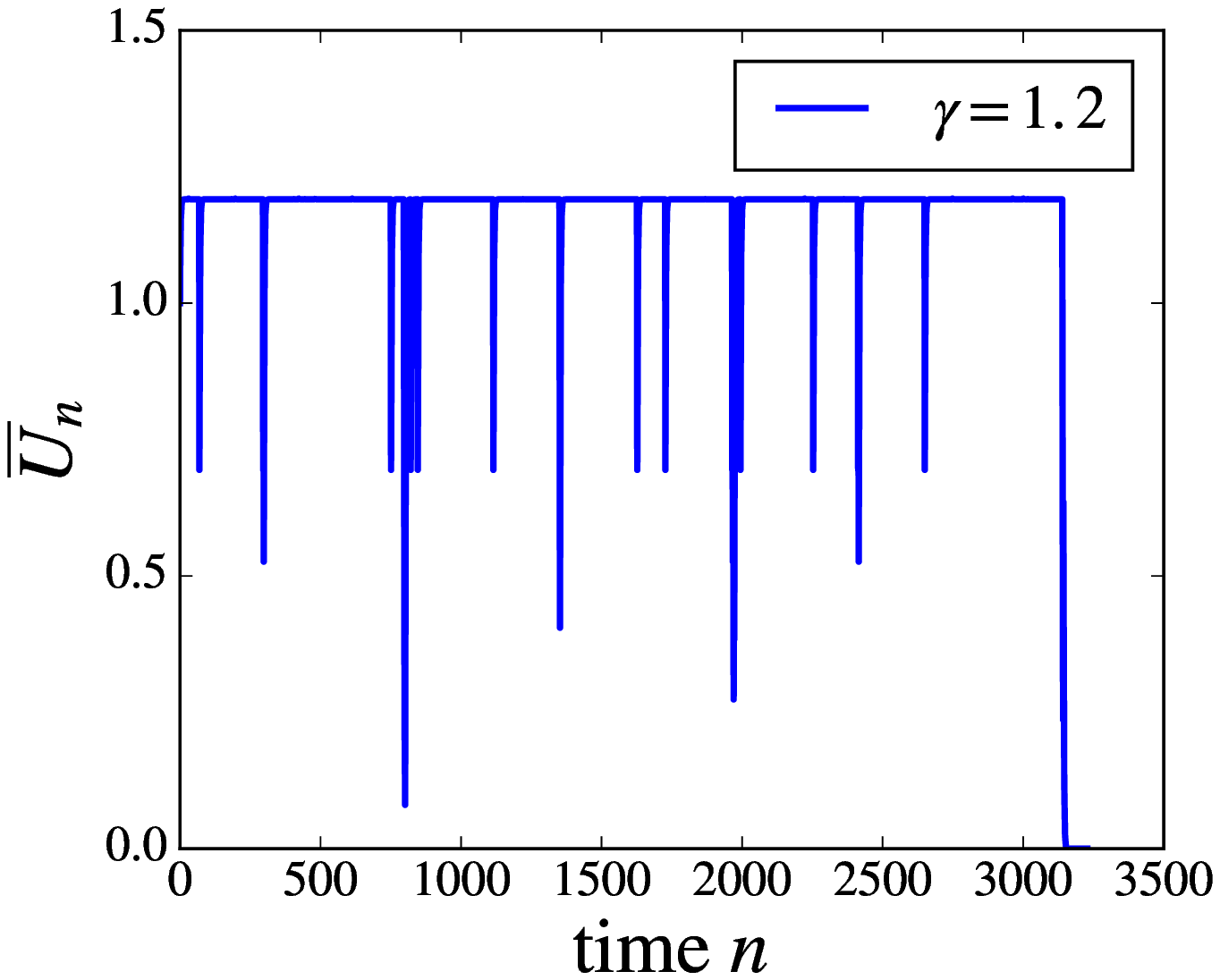}}
\makebox[\linewidth]{
\includegraphics[width=0.5\columnwidth]{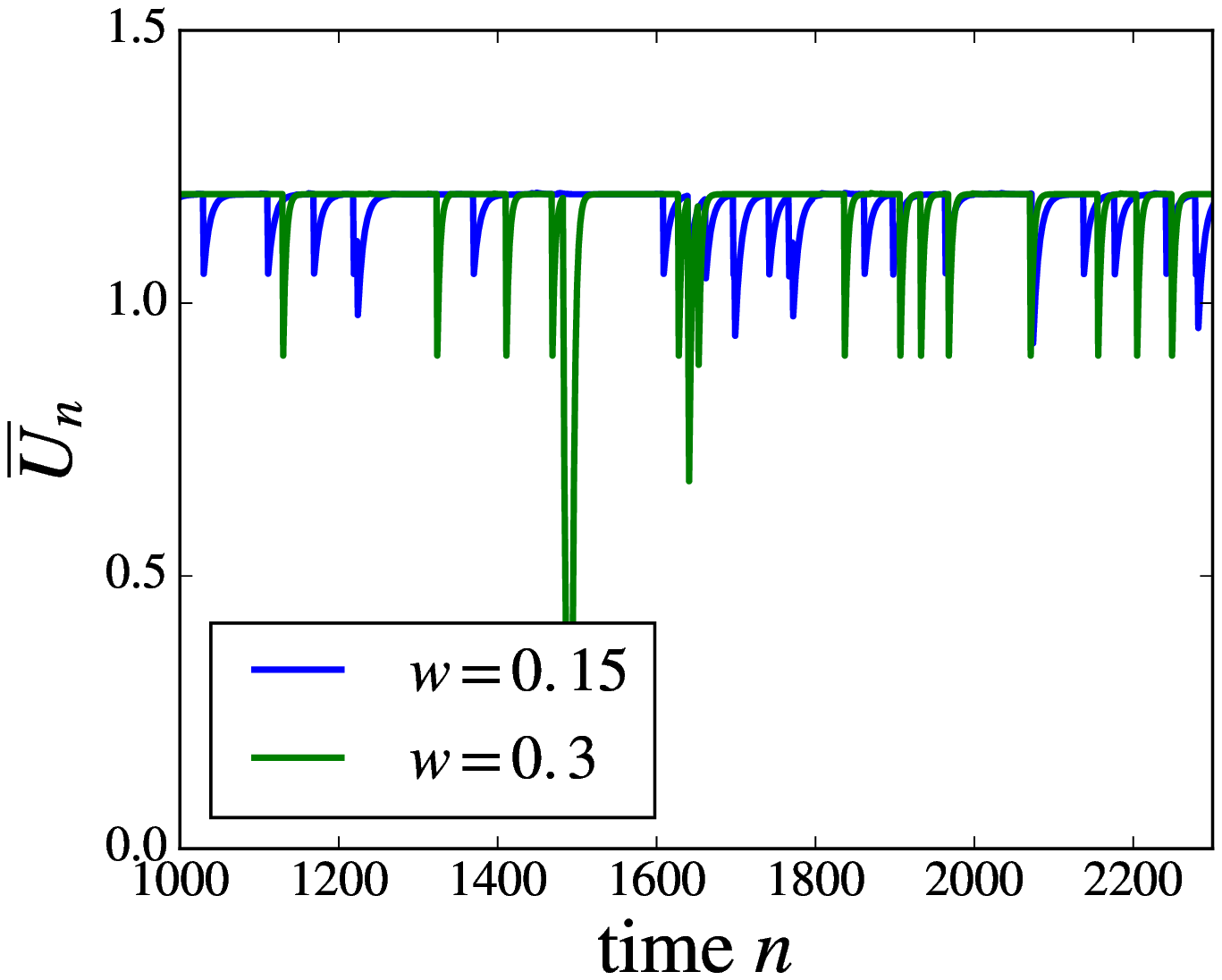}
\includegraphics[width=0.5\columnwidth]{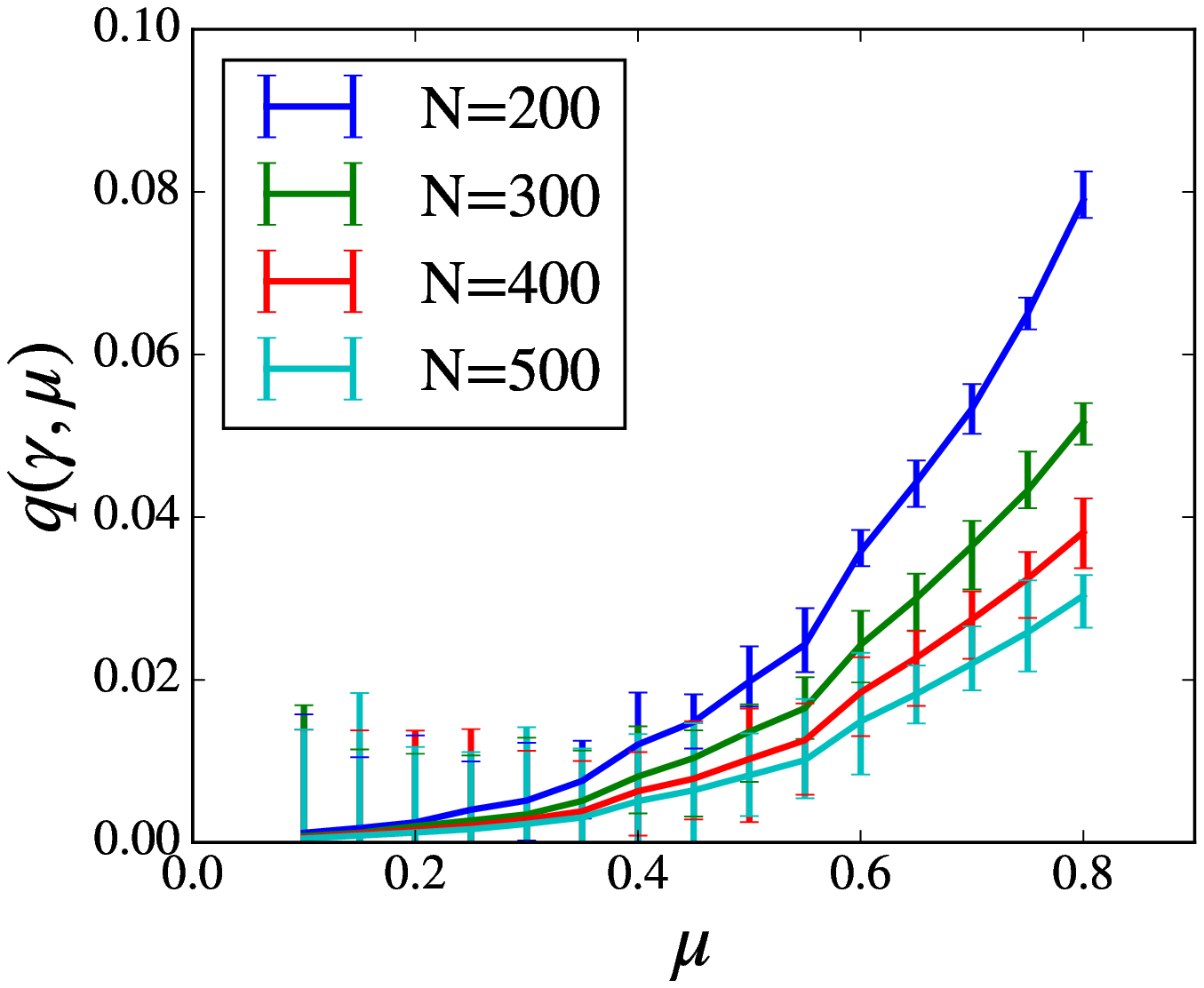}}
\captionof{figure}{{\bf(A)} Initial evolution of the mean potential from uniform initial condition $U_0(i)=1$ for all $i=1..N$, with fixed $w=0.5$ in the supercritical case -- blue line -- with  $\gamma=1.2$ ($\mu=0.583$) and subcritical case -- red line -- with $\gamma=0.8$ ($\mu=0.375$).  {\bf(B)} Complete time series of mean potential for the blue line in (A).
{\bf(C)} Example of two processes with $\gamma=1.2$ and $(w,\mu)=(0.15,0.878)$ and $(0.3,0.7533)$ yielding $m=2$ in the former case and $m=1$ in the latter. 
 {\bf(D)} proportion of time the system spends away from $\gamma$ as a function of $\mu$ for fixed $\gamma=1.2$, for several choices of $N$.}
\label{SampleEvol}
\end{minipage}

\subsection{Empirical extinction time}

\subsubsection{Stopping criterion} \label{StopCriterion}
Determining the extinction time of the process poses and immediate problem in  numerical experimentation because the last time the system fires does not configure a stopping time. In practice we wish to determine with great confidence the moment the system reaches a state of progressively smaller potentials that will prevent the system from ever firing again. To do so, we keep each simulation running until it reaches our stopping criterion time $n_s$

\[
n_s=\min\{n>0:\sum_{i\in I}{U_n(i)<u_c}\}.
\]

By the time the system reaches $n_s$ the probability it will ever fire again is given by

\[
P(\sum_{n=n_s}^{\infty} Y_n >0)=1-\sum_{j=0}^\infty (1-\mu^j u_c) \approx 1-\exp\lbrace-\frac{u_c}{1-\mu}\rbrace.
\]

We use $u_c=10^{-30}$ such that the probability that the system will ever fire again after $n_{s}$ is numerically immaterial for any reasonable choice of $\mu$.

The empirical extinction time of the process is then defined as the last time the system fired before $n_{s}$

\[
\tau_{obs}=\max\{n<n_s: Y_n=1 \}.
\]

Simulation routines were implemented in C++ using the pseudo-random number engine Mersenne Twister (mt19937\_64) of \emph{random} C++ library with machine clock as seed.

\subsubsection{Extracting $m$ from $\tau_{obs}(N)$}

To overcome the difficulty of computing numerically the extinction  time, being of order $N^m$, in the $(w,\mu)$ plane,
we performed simulations for a specific value of $w$ and  how it changes with increasing $N$ and compare its result with theoretical predictions regarding its dependence on the number of components across different values of leakage.  We fixed  $w=0.8$ and obtained the empirical pmf $\tau_{obs}(N)$ for each $\mu$ by computing the extinction time of 1000 simulations for each choice of $N$ and $\mu$. Figure \ref{EmpExps} A shows the empirical average $\overline{\tau}_{obs}(N) $  with respect to $N$ for each $\mu$. Error bars represent the standard deviation of empirical averages calculated by separating simulations in 100 trials of 10 samples  each. 

For each leakage value $\mu$, we use a linear regression model relative to the logarithm of the quantities $N$ and $\overline{\tau}_{obs}(N) $ to estimate the slope ${m_{obs}}$ and its standard deviation and compare it to the theoretical value $m=\lceil \frac{\log(1/\gamma)}{\log \mu} \rceil $. We observe  in Fig. \ref{EmpExps}  A that the estimated ${m_{obs}}$ (adjacent table) are at most $5\%$ distinct from the theoretical $m$ value  for each $\mu$. Moreover, different values of $\mu$ that yield the same $m$ produce parallel lines. 

Figure \ref{EmpExps} B shows how closely ${m_{obs}}$ (red dots with error bars) follows theoretical prediction $m$ (black line) for a more thorough grid of values $\mu$ corresponding to $m$ values up to $3$. One can notice that ${m_{obs}}$ discretely deviates from the theoretical curve as $m$ increases. This can be explained by the effect observed in Fig. \ref{SampleEvol} D that shows that $q(\gamma,\mu)$ increases with $\mu$. For $N=200$ for instance, the proportion of time the system spends away from $\gamma$ before dying increases from $0.1\%$ for $\mu=0.06$ to $4\%$ for $\mu=0.67$. Nonetheless, the empirical values are still agreeable with the theoretical $m$, showing that in practice the power law $N^m$ leads to a good prediction of the extinction time, even for $q(\gamma,\mu)$ as big as $4\%$.

In order to obtain the empirical pdf, we first generate $\mathcal{T}_{\mu,N}$ which is the set of observed extinction times for a specific value of $\mu$, for $N$ fixed, for a certain number of simulations. Then, we  compute
\[
\widehat{p}(T<\tau_{obs}(N)<T+t_b)=\frac{ | \{ \tau_{obs} \in \mathcal{T_{\mu,N} }:T<\tau_{obs}<T+t_b \} | } { | \mathcal{T}_{\mu,N} | }
 \] 
Blue bars in \ref{EmpExps}  C depict the empirical pdf  $\widehat{p}(T<{\tau_{obs}}_{\mu}(N)<T+t_b)$  for an example value of $\mu=0.56$,  with bin size $t_b=39168$,  for 1000 simulations. We indicate in a red curve an exponential distribution of parameter $N^{-m}{c}_{\mu,w}$. Here ${c}_{\mu,w}= e^{-{b}_{\mu,w}}$, where ${b}_{\mu,w}$ is the estimated constant term of the aforementioned linear model.

\noindent%
\begin{minipage}{\linewidth}

\makebox[\linewidth]{
\includegraphics[width=1.1\columnwidth]{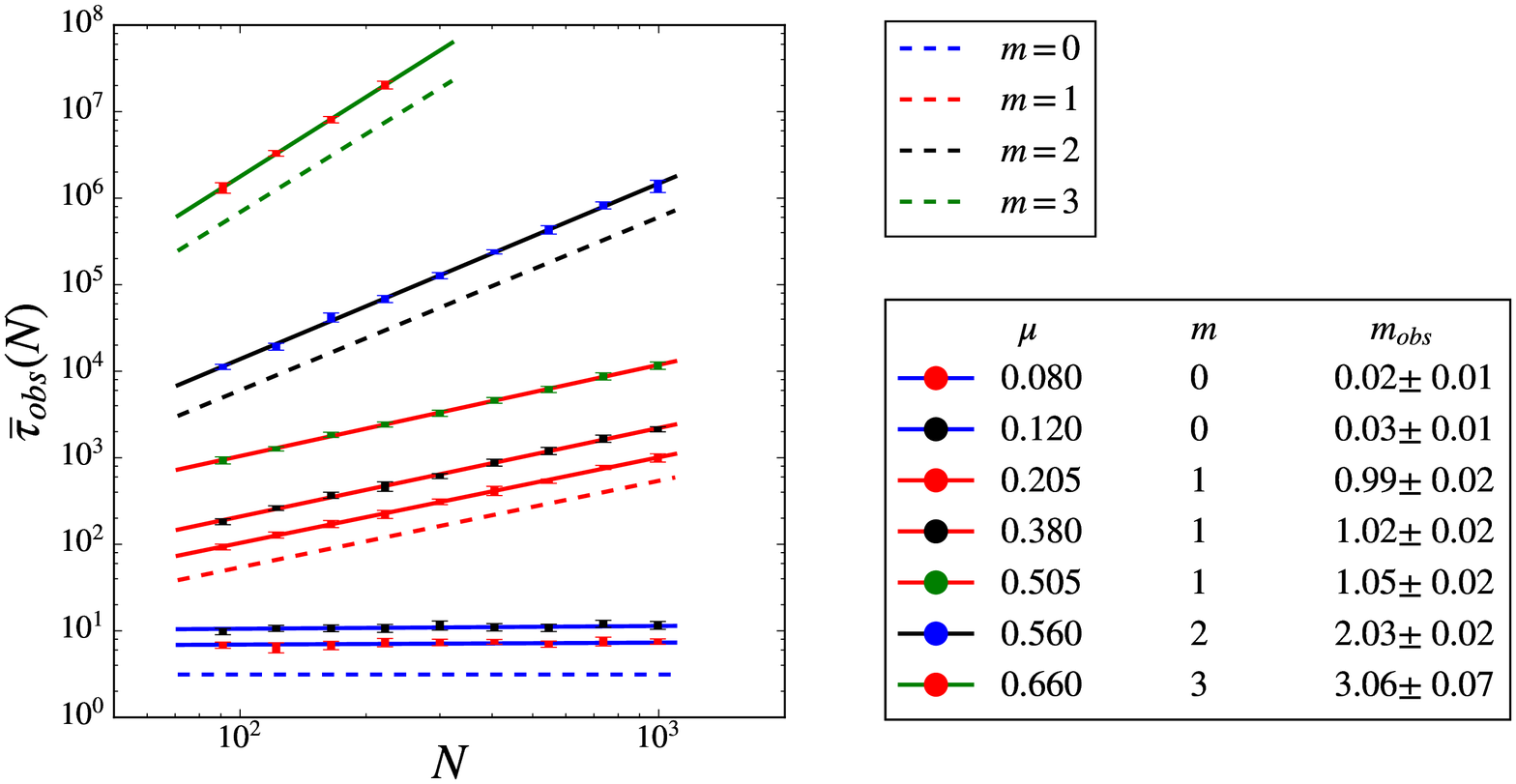}}

\makebox[\linewidth]{
\includegraphics[width=0.53\columnwidth]{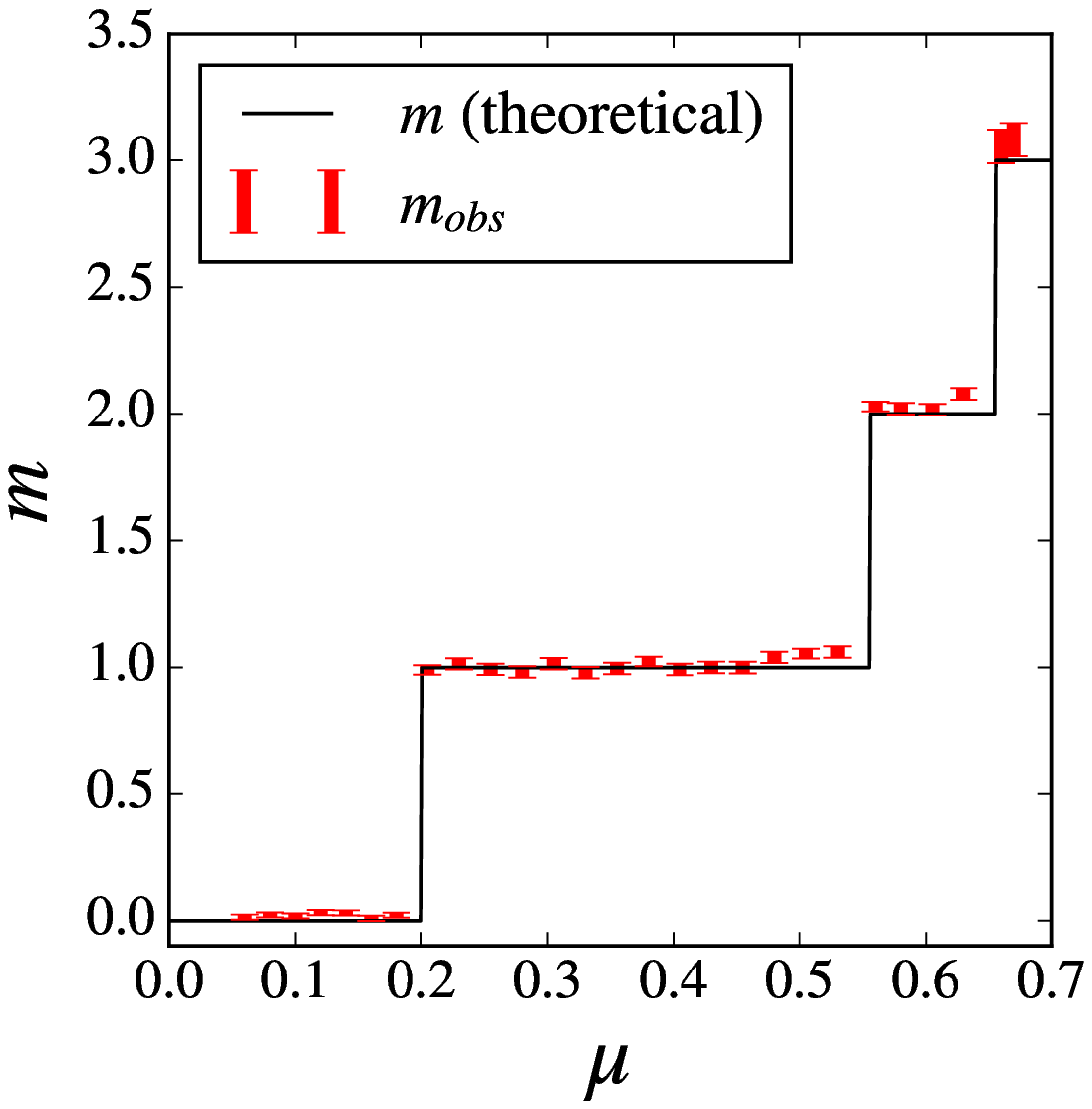}
\includegraphics[width=0.53\columnwidth]{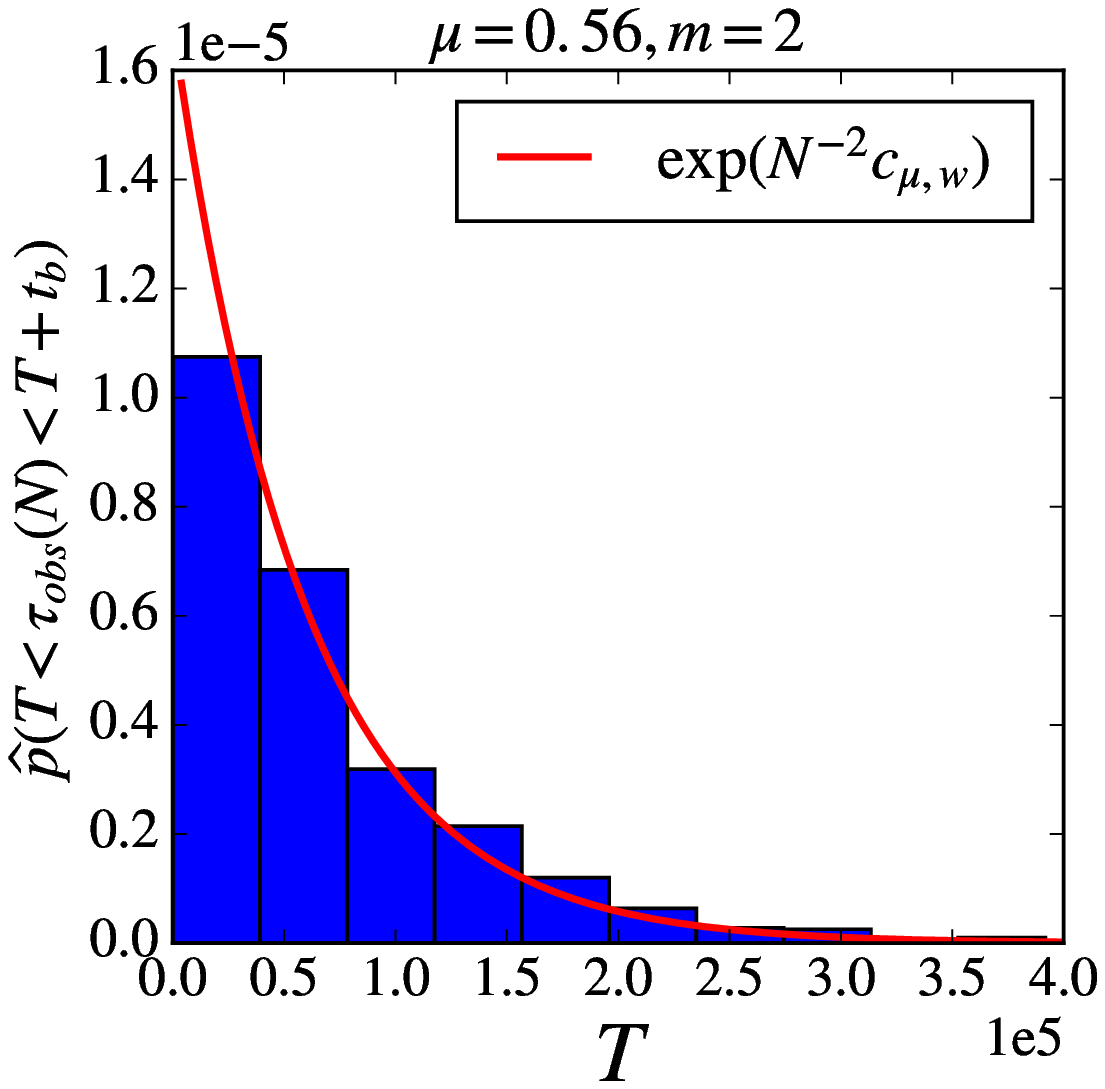}}

\captionof{figure}{\emph{Mean value of extinction time is proportional to   $N^m$.} Gain is fixed $w=0.8$. {\bf (A)} Fitted power law for each $\mu$ value with calculated ${m_{obs}}$ indicated in the adjacent table. Each line is colored according to the theoretical value of $m=\lceil \frac{\log(1/\gamma)}{\log \mu} \rceil $ for each $\mu$ value: red, blue, magenta and green indicate $m=0, 1, 2$ and $3$ respectively. Dotted colored lines are a guide to the eye to indicate the corresponding slope $m$. {\bf{(B)}} Red dots with errorbars indicate the estimated ${m_{obs}}$ with respect to $\mu$. Theoretical $m$ in black. {\bf{(C)}} Empirical pdf of $\tau_{obs}$ for $\mu=0.56$: $\widehat{p}(T<{\tau}_{obs}(N)<T+t_b)$. Red curve indicates the exponential distribution of parameter $N^{-m}{c}_{\mu,w}$. }\label{EmpExps}  
\end{minipage}

\bibliography {quasibib}
\bibliographystyle{plain}

\section*{Aknowledgements}

This article was produced as part of the activities of agreement FAPESP (SP-Brazil) and FCT (Portugal) with reference FAPESP/19805/2014
 and  of project  FAPESP Center for Neuromathematics (grant$\#2013/ 07699-0$ , S.Paulo Research Foundation). LB thanks FAPESP grant no 2016/24676-1. We thank A. Galves for important discussions on this work. 

\end{document}